\documentclass[submission,copyright,creativecommons]{eptcs}
\providecommand{\event}{AiML 2026} 

\usepackage{aiml26}
\usepackage{subcaption}
\usepackage{iftex}

\ifpdf
  \usepackage{underscore}         
  \usepackage[T1]{fontenc}        
\else
  \usepackage{breakurl}           
\fi

\usepackage{amsmath,amssymb,amsfonts,amsthm,latexsym,tikz}
\usetikzlibrary{shapes.geometric,positioning,backgrounds,patterns}
\usepackage{mathrsfs}
\usepackage{graphicx}
\usepackage{standalone}

\title{Stratified Counterpossible Logic}
\author{Chen Huang
\institute{Institute of Logic and Cognition\\Department of Philosophy\\Sun Yat-sen University\\ Guangzhou, China}
\email{huangch58@mail2.sysu.edu.cn}
\and
Xuefeng Wen
\institute{Institute of Logic and Cognition\\Department of Philosophy\\Sun Yat-sen University\\ Guangzhou, China}
\email{wxflogic@gmail.com}
}

\newcommand{\titlerunning}{Stratified Counterpossible Logic}
\newcommand{\authorrunning}{Chen Huang \& Xuefeng Wen}

\hypersetup{
  bookmarksnumbered,
  pdftitle    = {\titlerunning},
  pdfauthor   = {\authorrunning},
  pdfsubject  = {EPTCS},               
  pdfkeywords = {counterpossible conditionals, stratified semantics, impossible worlds, non-vacuism.} 
}

\begin{document}
\maketitle

\begin{abstract}
This paper presents two logic systems, $\mathbf{SCP}$ and its extension $\mathbf{SCP1}$, to distinguish between different types of impossibility. The semantics use a stratified structure that partitions worlds into logic-normal worlds ($N$) and anti-logic worlds ($I$). By defining a metaphysical accessibility relation within $N$, the systems separate metaphysical impossibility from logical contradiction. This allows for logical reasoning to be maintained even when dealing with metaphysical impossibilities. To address the challenge of vacuism, $\mathbf{SCP1}$ implements a non-empty constraint on the selection function, ensuring that counterpossibles with impossible antecedents are not trivially true but depend on the connection between the antecedent and the consequent. We provide proofs for the soundness, completeness, and decidability of both systems. Finally, we indicate the possibility of applying this stratified approach to other modal domains, such as deontic or epistemic logic.
\end{abstract}

\section{Introduction}

Counterfactual semantics typically aim to balance formal simplicity
with intuitive plausibility. While the Stalnaker-Lewis systems succeed
for standard counterfactuals \cite{Stalnaker1968,Lewis1973,Lewis1981},
they yield counterintuitive results for counterpossibles—conditionals
with impossible antecedents.
\begin{enumerate}
\item If Hobbes had squared the circle, sick children in Afghanistan would
have been cured. 
\item If Hobbes had squared the circle, we would have been surprised. 
\end{enumerate}
For Stalnaker and Lewis, both conditionals are true because squaring
the circle is impossible. Intuitively, however, the first seems false
and the second true. This conflict between theory and intuition motivates
the debate between vacuism and non-vacuism.

Vacuism asserts that all counterpossibles are vacuously true. While
Williamson \cite{williamson2018Counterpossibles} defends this view
for its logical strength, it has difficulty accommodating mathematical
proofs by contradiction or metaphysical thought experiments. If every
impossible antecedent yields a trivial truth, we cannot distinguish
substantive hypotheses from absurd contradictions.

Non-vacuism, in contrast, argues that some counterpossibles are false
\cite{brogaard2013RemarksCounterpossibles}. To evaluate impossible
antecedents within possible-world semantics, Nolan \cite{nolan1997ImpossibleWorldsModest}
introduced impossible worlds and the Strangeness of Impossibility
Condition (SIC): any possible world is more similar to the actual
world than any impossible world. This framework allows for the evaluation
of counterpossibles while preserving classical logic for possible
antecedents.\footnote{For further research on formal logical systems, see \cite{berto2018WilliamsonCounterpossibles,weissFrontiersConditionalLogic}.} Weiss \cite{weiss2017SemanticsCounterpossibles} evaluated the relative
merits of sphere semantics and ceteris paribus logic, highlighting
the potential of weak systems to provide a uniform semantic framework. French \cite{french2024ReflectionVacuismStrangeness} proposed a conciliatory
approach. By employing the translation $t(\varphi>\psi):=\Diamond t(\varphi)\to(t(\varphi)>t(\psi))$,
French proved that all theorems of vacuist logic are theorems of non-vacuist
logic under translation. This addresses the concern that non-vacuism
necessarily weakens the logic by abandoning valid inference rules.

While these works progress in demarcating the possible from the impossible,
they do not sufficiently model the internal structure of the impossible
domain. Existing frameworks often conflate metaphysical impossibilities
(e.g., ``water is not $H_{2}O$'') with logical contradictions (e.g.,
$P\land\neg P$). This leads to a trade-off: to accommodate contradictions,
systems often sacrifice inference rules that are still needed for
reasoning about metaphysical hypotheses.

This paper argues for both a distinction between the possible and
impossible and a stratification of the impossible realm.\footnote{Similar stratified approaches exist; for example, Gillies \cite{gillies2007CounterfactualScorekeeping}
used scorekeeping to establish stratified strict conditional semantics
for context-sensitivity.} We propose Stratified Counterpossible Logic ($\mathbf{SCP}$), which
partitions worlds into logic-normal worlds ($N$) and anti-logical
worlds ($I$). By nesting a metaphysical accessibility relation ($R_{M}$)
within $N$, we create a three-layered structure that preserves logical
reasoning even when metaphysical necessities are violated.

\section{Philosophical Motivation}

Existing literature on non-vacuism (e.g., \cite{nolan1997ImpossibleWorldsModest,weissFrontiersConditionalLogic})
tends to adopt a binary perspective: on one side are possible worlds
that obey all rules, and on the other are impossible worlds that break
them. However, this approach implicitly assumes that all impossibilities
have the same logical status. Treating all impossibilities identically
forces logical systems into a dilemma when dealing with counterpossible
reasoning.

To characterize these differences, this paper adopts classical logic
as the base logic to define logically possible and logically impossible
worlds. \footnote{While many non-classical logics, such as paraconsistent or intuitionistic
logic, offer deep philosophical insights, we leave them for future
research to avoid overcomplicating the semantic model.} The logical space is partitioned into two regions:

\textbf{Logic-normal worlds ($N$):} This region contains all worlds
that are consistent under the base logic (classical logic). It includes
not only the actual world but also ``anti-essential'' worlds that
violate metaphysical laws yet still follow classical inference rules,
such as the Law of Excluded Middle and the Law of Non-Contradiction.

\textbf{Anti-logical worlds ($I$):} This region contains situations
that violate the norms of the base logic.

In this framework, worlds closed under non-classical logics, such
as paraconsistent worlds where $\varphi\land\neg\varphi$ holds, are
categorized as ``logically impossible'' relative to the base $N$.
While $I$ can accommodate diverse non-classical structures, this
paper treats it primarily as a domain where base logical norms fail.

Within the logically normal domain $N$, we distinguish between metaphysical
possibility and impossibility via the relation $R_{M}$. The actual
world and its $R_{M}$-accessible worlds constitute the metaphysically
possible region. Conversely, $R_{M}$-inaccessible worlds within $N$
represent ``anti-essential'' scenarios, such as ``water is XYZ''
or ``gold has atomic number 50.''

This stratified $N/R_{M}/I$ structure echoes philosophical practice
in scientific and metaphysical thought experiments. When assuming
antecedents that violate natural laws or essential properties, we
typically continue to reason using classical logic. This suggests
that such anti-metaphysical scenarios should remain within the logic-normal
domain $N$. Consider the following example:

\begin{example} If gold had atomic number 50, gold would still be
a metal. \end{example}

Metaphysically, gold necessarily and essentially has atomic number
79; thus, ``the atomic number of gold is 50'' is a metaphysical
impossibility. In unstratified counterpossible logics, because the
antecedent is impossible, it must point to an ``impossible world.''
Since impossible worlds are generally unconstrained by logic, one
might lose the inference that ``gold is a metal,'' which is based
on chemical classification. Intuitively, even if we imagine a change
in the atomic number of gold, we are still reasoning with classical
logic. Therefore, such ``anti-essential'' scenarios should be retained
within the logic-normal worlds ($N$).

\section{Stratified Counterpossible Logic \texorpdfstring{($\mathbf{SCP}$)}{(SCP)}}

This section presents Stratified Counterpossible Logic (henceforth, $\mathbf{SCP}$). The core intuition of $\mathbf{SCP}$ is to partition
the logical space into three levels: metaphysically possible worlds,
logically possible (but metaphysically impossible) worlds, and anti-logical
worlds.

$\mathbf{SCP}$ employs equivalence relations to characterize metaphysical
and logical possibilities, presupposing that metaphysical possibility
implies logical possibility.

\begin{definition}[Formal Language] Given a set of atomic propositions
$P$, the formulas of the language $\mathscr{L}$ are generated by
the following BNF: 
\[
\varphi::=p\mid\neg\varphi\mid(\varphi\to\varphi)\mid\Box_{M}\varphi\mid\Box_{L}\varphi\mid(\varphi>\varphi),
\]
where $p\in P$. In this paper, contradictions and tautologies refer
to those in classical logic, denoted by $\bot$ and $\top$ respectively.
\end{definition}

\subsection{Semantics}

We first provide the model-theoretic definition of $\mathbf{SCP}$.
The model is based on selection semantics for conditionals and introduces
stratified modal accessibility relations to characterize the inclusion
relationship between metaphysical and logical possibilities.

\begin{definition} An $\mathbf{SCP}$ model is a tuple $M=(W,N,R_{M},R_{L},f,V)$,
where: \label{def:scp-model} 
\begin{itemize}
\item $W$ is a non-empty set of worlds (including impossible worlds);
\item $N$ is a subset of $W$, consisting of normal worlds that follow the laws of classical
logic;
\item $I:=W-N$ is thus a set of anti-logical worlds that are not constrained
by classical logic;
\item $R_{M},R_{L}\subseteq N\times N$ are binary equivalence relations
defined on $N$, representing the metaphysical accessibility relation
and the logical accessibility relation, respectively; 
\item $f:N\times\mathscr{L}\to\mathcal{P}(W)$ is a selection function;
\item $V:\mathscr{L}\to\mathcal{P}(W)$ is a valuation function. 
\end{itemize}
\end{definition}
Note that for brevity $V$ is not merely defined on atoms, but on all formulas.

\begin{definition}[Truth Conditions for $\mathbf{SCP}$] The truth
relation $M,w\Vdash\varphi$ for a formula $\varphi$ at world $w$
in model $M$ is defined recursively as follows: \label{def:scp-truth}
\end{definition}
\begin{itemize}
\item In normal worlds ($w\in N$): 
\begin{itemize}
\item For an atomic proposition $p$, $M,w\Vdash p\iff w\in V(p)$. 
\item $M,w\Vdash\neg\varphi\iff M,w\nVdash\varphi$. 
\item $M,w\Vdash\varphi\to\psi\iff M,w\nVdash\varphi$ or $M,w\Vdash\psi$. 
\item $M,w\Vdash\Box_{M}\varphi\iff R_{M}(w)\subseteq[\varphi]$, where
$[\varphi]:=\{w\in W\mid M,w\Vdash\varphi\}$. 
\item $M,w\Vdash\Box_{L}\varphi\iff R_{L}(w)\subseteq[\varphi]$. 
\item $M,w\Vdash\varphi>\psi\iff f(w,\varphi)\subseteq[\psi]$. 
\end{itemize}
\item In anti-logical worlds ($w\in I$): 
\begin{itemize}
\item For any formula $\varphi\in\mathscr{L}$: $M,w\Vdash\varphi\iff w\in V(\varphi)$. 
\end{itemize}
\end{itemize}
\begin{remark} In $I$, the truth values of complex formulas (e.g.,
$\varphi\land\psi$) are no longer defined recursively. This allows
for logical impossibilities, such as $\varphi$ being true while $\varphi\lor\psi$
is false, or $\varphi\land\neg\varphi$ being true. \end{remark}

\begin{definition}[Semantic Constraints] For any $w\in N$, the
model $M$ must satisfy the following constraints: \label{def:scp-constraints}
\[
\begin{array}{@{}ll@{\qquad}ll@{}}
\text{(modal)} & R_{M}(w)\subseteq R_{L}(w) & 
\text{(soc-m)} & \text{If }w\in[\Diamond_{M}\varphi],\ \text{then }f(w,\varphi)\subseteq R_{M}(w) \\[4pt]
\text{(id)} & f(w,\varphi)\subseteq[\varphi] & 
\text{(sic)} & \text{If }w\in[\Diamond_{L}\varphi],\ \text{then }f(w,\varphi)\subseteq N
\end{array}
\]
\end{definition}

\begin{figure}[htbp]
    \centering
        \centering
        \resizebox{0.5\textwidth}{!}{
            \begin{tikzpicture}[
                outerworld/.style={rectangle, draw, minimum width=10cm, minimum height=7.5cm, ultra thick},
                normalzone/.style={ellipse, draw, fill=gray!5, minimum width=7.5cm, minimum height=5.5cm, thick},
                rlcluster/.style={ellipse, draw, dashed, blue!70, fill=blue!5, minimum width=3.5cm, minimum height=2.5cm, thick},
                rmclass/.style={circle, draw, dotted, black, fill=white, minimum size=1.2cm, inner sep=0pt}
            ]
                \node[outerworld] (W) at (0,0) {};
                \node[anchor=north west, font=\huge\bfseries] at (W.north west) {$W$};

                \node[font=\huge] at (-4, 3) {$I$};
                \node[font=\large, text=gray] at (-3.1, 2.6) {(Impossible Worlds)};

                \node[normalzone] (N) at (0,0) {};
                \node[anchor=north, font=\Large\bfseries] at (N.north) {$N$};

                \node[rlcluster] (RL1) at (-1.6, 0.3) {};
                \node[rlcluster] (RL2) at (1.7, -0.7) {};
                \node[blue, font=\large\bfseries] at (-1.5, 1.9) {$R_L$-cluster 1};
                \node[blue, font=\large\bfseries] at (1.8, 0.7) {$R_L$-cluster 2};

                \node[rmclass] at (-2.2, 0.4) {\large $R_M$};

                \node[rmclass] at (-0.9, -0.1) {\large $R_M$};
                \node[rmclass] at (1.0, -0.8) {\large $R_M$};
 
                \node[rmclass] at (2.6, -1.1) {\large $R_M$};

                \node[draw, fill=white, shift={(-0.2,0.2)}, anchor=south east] at (W.south east) {
                    \large
                    \begin{tabular}{ll}
                    \tikz[baseline=-0.5ex]\draw[dashed, blue, ultra thick] (0,0) -- (0.6,0); &  $R_L$ \\
                    \tikz[baseline=-0.5ex]\draw[dotted, black, ultra thick] (0,0) -- (0.6,0); &  $R_M$ \\
                    \end{tabular}
                };
            \end{tikzpicture}
        }
        \caption{The Stratified World Model of the $\mathbf{SCP}$ System}
        \label{fig:scpmodel0}
\end{figure}
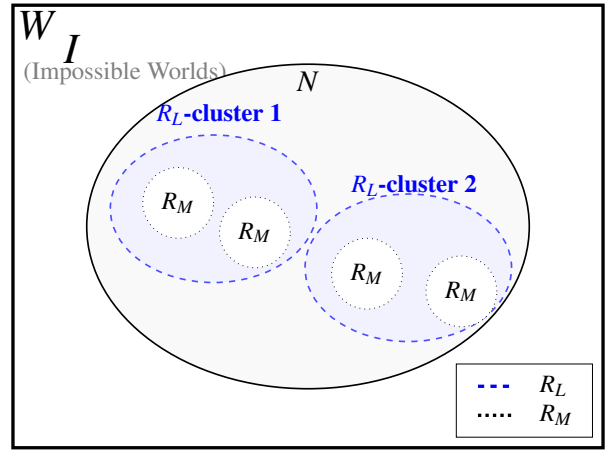
Figure \ref{fig:scpmodel0} illustrates the hierarchical structure
of the $\mathbf{SCP}$ system's world model. Its semantic framework
consists of the following four core components:

\textbf{The Universe $W$ and Impossible Worlds $I$}: The outermost
rectangle represents the total set $W$ of all logically possible
and impossible worlds. Points within $I$ represent anti-logical worlds
where logical laws may fail, providing a semantic buffer for handling
counterpossibles. 

\textbf{Normal World Region $N$}: The central oval region $N$ represents
the set of normal worlds that satisfy all logical axioms of the $\mathbf{SCP}$
system. 

\textbf{Logical Equivalence Classes $R_{L}$}: Within $N$, blue dashed
lines partition the space into $R_{L}$-clusters. Disjoint equivalence
classes represent different boundaries of logical possibility. 

\textbf{Metaphysical Equivalence Classes $R_{M}$}: Within each $R_{L}$
equivalence class, $R_{M}$ classes are further partitioned. This
reflects the nested property $R_{M}\subseteq R_{L}$, indicating that
metaphysical possibility is a finer-grained division restricted by
logical possibility. 

\begin{definition}[Validity] Given a class $\mathsf{S}$ of $\mathbf{SCP}$ models, an inference from a set of formulas $\Gamma$ to $\varphi$
is valid in $\mathsf{S}$, denoted by $\Gamma\vDash_{\mathsf{S}}\varphi$,
if and only if for any $M=(W,N,R_{M},R_{L},f,V)\in\mathsf{S}$ and
any $w\in N$, if $M,w\Vdash\Gamma$, then $M,w\Vdash\varphi$. \end{definition}

Validity is strictly restricted to the set of normal
worlds $N$. This aligns with the standard approach in counterpossible
logic: we evaluate hypotheses about the ``impossible'' from the
standpoint of worlds where logical laws remain in effect. Let $\mathsf{SCP}$
denote the class of models satisfying all semantic constraints. In
this section, $\vDash_{\mathsf{SCP}}$ is often abbreviated as $\vDash$.

\subsection{Validities of \texorpdfstring{$\mathbf{SCP}$}{SCP}}

\label{subsec:scp-validities}

We examine the key logical properties of $\mathbf{SCP}$. In standard counterfactual semantics (such as Lewis-Stalnaker semantics),
conditionals with impossible antecedents are treated as ``vacuously
true.'' Non-vacuism rejects this, arguing that even if an antecedent
is unsatisfiable, the conditional should have non-trivial truth conditions.
In other words, a contradiction or an impossibility should not imply
everything. Consequently, a counterpossible logic should not validate
the following two formulas:\footnote{For a more detailed classification of counterpossible logics, see
\cite{frenchCLASSICALCOUNTERPOSSIBLES}.} 
\begin{center}
(EPQ)\quad{}$\bot>\varphi$ \qquad{}\qquad{}(VAC)\quad{}$\neg\Diamond\varphi\to(\varphi>\psi)$ 
\par\end{center}

\begin{proposition} $\mathbf{SCP}$ is a counterpossible logic. 
\[
\begin{array}{@{}lll@{}}
1. \ \nvDash\bot>\varphi\qquad & 
2. \ \nvDash\neg\Diamond_{L}\varphi\to(\varphi>\psi)\qquad & 
3. \ \Diamond_{L}\varphi\vDash\varphi>\top
\end{array}
\]
\end{proposition}

Regarding the first point, in $\mathbf{SCP}$, $f(w,\bot)\subseteq I$.
Since truth values are assigned arbitrarily in anti-logical worlds,
$\bot>\varphi$ is not universally valid. The same reasoning applies
to the second point. The third point indicates that $\mathbf{SCP}$
does not weaken logical strength in all cases. As long as the antecedent
is logically possible ($\Diamond_{L}\varphi$), the conditional operator
reverts to classical semantic rules, thereby ensuring the necessity
of the tautology $\top$.

\begin{proposition} The following hold. \label{prop:necessity-LM} 
\[
\begin{array}{@{}lll@{}}
1. \ \nvDash\Box_{L}(\Box_{M}\varphi\to\varphi) & \qquad
2. \ \Box_{M}\psi\nvDash\Diamond_{L}\varphi\to(\varphi>\psi) & \qquad
3. \ \vDash\Box_{L}\varphi\to\Box_{M}\Box_{L}\varphi
\end{array}
\]
\end{proposition}
\begin{proof}
We only prove the Item 3. Suppose $M,w\Vdash\Box_{L}\varphi$. Since $R_{L}$ is an equivalence
relation, for any $v\in R_{L}(w)$, we have $M,v\Vdash\Box_{L}\varphi$.
Given $R_{M}(w)\subseteq R_{L}(w)$, it follows that $M,w\Vdash\Box_{M}\Box_{L}\varphi$.
\end{proof}

Item 1 shows that it is logically possible for a metaphysically
necessary proposition to be false. There exists a logic-normal world
$w\in N$ (logically possible relative to the actual world but metaphysically
impossible) that is logically accessible, yet metaphysical laws fail
in some worlds accessible from it. This ensures that we can discuss
scenarios like ``water is not $H_{2}O$'' without causing logical
collapse.

Item 2 further indicates that metaphysically necessary truths
are not necessarily preserved under anti-essential hypotheses. For
instance, let $\psi$ be ``water is $H_{2}O$'' (metaphysically
necessary) and $\varphi$ be ``water is $XYZ$'' (anti-essential
but logically possible). This is a distinct advantage of the $\mathbf{SCP}$
system: it allows for a meaningful discussion of the consequences
of ``water not being $H_{2}O$'' without the derivation becoming
vacuous due to the necessity of $\psi$.

Item 3 reflects the ``overlap'' of the logical hierarchy
over the metaphysical one. While metaphysical laws can be false at
the logical level (see the first point), logical laws must be true
at the metaphysical level. Once a proposition $\varphi$ is judged
logically necessary (i.e., true in all worlds in $R_{L}$), it remains
logically necessary regardless of how we shift our metaphysical perspective
via $R_{M}$.

\subsection{Axiomatic system}

The axiomatic system for $\mathbf{SCP}$ is built
upon classical propositional logic, supplemented by $\mathbf{S5}$ modal operators
$\Box_{L}$ and $\Box_{M}$. The system is defined as follows: \label{def:SCP-axioms}

\medskip
\noindent\textbf{Basic Propositional and Modal Axioms} 
\[
\begin{array}{@{}ll@{\qquad}ll@{}}
\text{(PC)} & \text{All tautologies of propositional logic}\\
\text{(KL)} & \Box_{L}(\varphi\to\psi)\to(\Box_{L}\varphi\to\Box_{L}\psi) & \text{(KM)} & \Box_{M}(\varphi\to\psi)\to(\Box_{M}\varphi\to\Box_{M}\psi)\\
\text{(TL)} & \Box_{L}\varphi\to\varphi & \text{(TM)} & \Box_{M}\varphi\to\varphi\\
\text{(4L)} & \Box_{L}\varphi\to\Box_{L}\Box_{L}\varphi & \text{(4M)} & \Box_{M}\varphi\to\Box_{M}\Box_{M}\varphi\\
\text{(5L)} & \Diamond_{L}\varphi\to\Box_{L}\Diamond_{L}\varphi & \text{(5M)} & \Diamond_{M}\varphi\to\Box_{M}\Diamond_{M}\varphi
\end{array}
\]

\noindent\textbf{Modal and Conditional Axioms} 
\begin{align*}
 & \text{(MODAL)} &  & \Box_{L}\varphi\to\Box_{M}\varphi\\
 & \text{(ID)} &  & \varphi>\varphi\\
 & \text{(SOC-M)} &  & \Diamond_{M}\varphi\land\Box_{M}\psi\to(\varphi>\psi)
\end{align*}

\noindent\textbf{Inference Rules} 
\begin{gather*}
\frac{\varphi,\quad\varphi\to\psi}{\psi}\text{(MP)}\qquad\quad\frac{\varphi}{\Box_{L}\varphi}\text{(RN}_{L})\qquad\quad\frac{\varphi}{\Box_{M}\varphi}\text{(RN}_{M})\\[10pt]
\frac{(\varphi_{1}\wedge\ldots\wedge\varphi_{n})\to\varphi}{\Diamond_{L}\psi\to(((\psi>\varphi_{1})\wedge\ldots\wedge(\psi>\varphi_{n}))\to(\psi>\varphi))}\text{(SIC)}
\end{gather*}

\subsection{Soundness and completeness}

\begin{theorem}[Soundness of $\mathbf{SCP}$] For any $\Gamma\cup\{\varphi\}\subseteq\mathscr{L}$,
if $\Gamma\vdash_{SCP}\varphi$, then $\Gamma\vDash_{SCP}\varphi$.
\label{theorem:scp-soundness} 
\end{theorem}
\begin{proof}
The validity of MODAL, ID, PC, and S5-related axioms and rules is
straightforward.

(SIC): Suppose $\vDash(\varphi_{1}\wedge...\wedge\varphi_{n})\to\varphi$.
Let $M,w\Vdash\Diamond_{L}\psi$ and $M,w\Vdash\psi>\varphi_{1},\dots,M,w\Vdash\psi>\varphi_{n}$.
According to the truth conditions and the (sic) constraint, $f(w,\psi)\subseteq[\varphi_{1}],\dots,f(w,\psi)\subseteq[\varphi_{n}]$
and $f(w,\psi)\subseteq N$. Thus, $f(w,\psi)\subseteq[\varphi_{1}\wedge...\wedge\varphi_{n}]$.
Since $\vDash(\varphi_{1}\wedge...\wedge\varphi_{n})\to\varphi$,
for any world $x\in N$, $x\in[\varphi_{1}\wedge...\wedge\varphi_{n}]\implies x\in[\varphi]$,
meaning $[\varphi_{1}\wedge...\wedge\varphi_{n}]\subseteq[\varphi]$.
Therefore, $M,w\Vdash\psi>\varphi$.

(SOC-M): Suppose $w\in N$ and $M,w\Vdash\Diamond_{M}\varphi\land\Box_{M}\psi$.
Since $M,w\Vdash\Box_{M}\psi$, it follows that $R_{M}(w)\subseteq[\psi]$.
Given $M,w\Vdash\Diamond_{M}\varphi$, the constraint (soc-m) implies
$f(w,\varphi)\subseteq R_{M}(w)$. Therefore, $M,w\Vdash\varphi>\psi$.
\end{proof}

To prove the completeness of the $\mathbf{SCP}$ system, we use the
canonical model construction. The core objective is to show that every
consistent set of formulas is satisfiable in some model. The key challenge in the construction is handling the opposition between
logic-normal worlds and anti-logical worlds. We must construct maximal
consistent sets (MCS) to represent normal worlds, and all subsets
of the language to serve as anti-logical worlds that accommodate logical
impossibilities.

\begin{definition}[Canonical Models for $\mathbf{SCP}$]
The canonical model $M^{c}=(W^{c},N^{c},f^{c},R_{M}^{c},R_{L}^{c},V^{c})$
is defined as follows: \label{def:canonical-model} 
\begin{itemize}
\item $W^{c}=\mathcal{P}(\mathscr{L})$ 
\item $N^{c}=\{w\mid w\text{ is an MCS of }\mathbf{SCP}\}$ with $I^{c}:=W^{c}-N^{c}$.
\item $\forall x,y\in N^{c},xR_{M}^{c}y$ iff $\{\varphi\mid\Box_{M}\varphi\in x\}\subseteq y$. 
\item $\forall x,y\in N^{c},xR_{L}^{c}y$ iff $\{\varphi\mid\Box_{L}\varphi\in x\}\subseteq y$. 
\item $f^{c}(w,\varphi)=\begin{cases}
\left\{ x\in N^{c}\mid\left\{ \psi\mid\varphi>\psi\in w\right\} \subseteq x\right\}  & \text{if }\Diamond_{L}\varphi\in w\\
\left\{ x\in W^{c}\mid\left\{ \psi\mid\varphi>\psi\in w\right\} \subseteq x\right\}  & \text{if }\Diamond_{L}\varphi\notin w
\end{cases}$ 
\item $V^{c}(\varphi)=\{\Delta\mid\varphi\in\Delta,\Delta\in\mathcal{P}(\mathscr{L})\}$. 
\end{itemize}
Let $[[\varphi]]:=\{w\in W^{c}\mid M^{c},w\Vdash\varphi\}$. \end{definition}

\begin{lemma}[Truth Lemma] For any $w\in W^{c}$, $M^{c},w\Vdash\alpha$
iff $\alpha\in w$. \label{lem:truth-lemma} \end{lemma}
\begin{proof}
For any $w\in I^{c}$ and $\alpha\in\mathscr{L}$, the truth condition
gives $M^{c},w\Vdash\alpha\iff w\in V^{c}(\alpha)\iff\alpha\in w$.

For any $w\in N^{c}$, we proceed by induction on the complexity of
formula $\alpha$. The cases for atomic propositions and Boolean connectives
are trivial.

Case $\alpha=\Box_{M}\varphi$:

$(\Leftarrow)$ If $\Box_{M}\varphi\in w$, then for all $x\in R_{M}^{c}(w)$,
$\varphi\in x$. By IH, $M^{c},x\Vdash\varphi$
for all $x\in R_{M}^{c}(w)$, hence $M^{c},w\Vdash\Box_{M}\varphi$.

$(\Rightarrow)$ If $\Box_{M}\varphi\notin w$, we construct a set
$\Gamma=\{\psi\mid\Box_{M}\psi\in w\}\cup\{\neg\varphi\}$. If $\Gamma$
were inconsistent, then $\vdash(\psi_{1}\wedge\dots\wedge\psi_{n})\to\varphi$
for some $\Box_{M}\psi_{i}\in w$. By normality (K and RN), $\vdash(\Box_{M}\psi_{1}\wedge\dots\wedge\Box_{M}\psi_{n})\to\Box_{M}\varphi$,
implying $\Box_{M}\varphi\in w$, a contradiction.

Case $\alpha=\Box_{L}\varphi$: Similar to the $\Box_{M}\varphi$
case.

Case $\alpha=\varphi>\psi$:

$(\Leftarrow)$ If $\varphi>\psi\in w$, then by the definition of
$f^{c}$, for any $x\in f^{c}(w,\varphi)$, we have $\psi\in x$.
By IH, $M^{c},x\Vdash\psi$, so $M^{c},w\Vdash\varphi>\psi$.

$(\Rightarrow)$ Suppose $\varphi>\psi\notin w$.

If $\Diamond_{L}\varphi\notin w$: Let $u=\{\chi\mid(\varphi>\chi)\in w\}$.
Since $\varphi>\psi\notin w$, $\psi\notin u$. If $u\in N^{c}$,
by IH, $M^{c},u\nVdash\psi$. Since $u\in f^{c}(w,\varphi)$,
$M^{c},w\nVdash\varphi>\psi$. If $u\notin N^{c}$, by the truth condition,
we also have $M^{c},w\nVdash\varphi>\psi$.

If $\Diamond_{L}\varphi\in w$: We show that $\{\chi\mid(\varphi>\chi)\in w\}\cup\{\neg\psi\}$
is consistent. If not, there exist $\chi_{i}$ such that $\vdash(\chi_{1}\land\dots\land\chi_{n})\to\psi$.
By the SIC rule, $\vdash\Diamond_{L}\varphi\to(((\varphi>\chi_{1})\wedge\ldots\wedge(\varphi>\chi_{n}))\to(\varphi>\psi))$.
This implies $\varphi>\psi\in w$, a contradiction. 
\end{proof}

We further prove that $M^{c}$ satisfies the semantic constraints
of the defined class of models.

\begin{lemma}\label{lem:c-constraints-scp} The canonical model $M^{c}$
satisfies the semantic constraints: for any $w\in N^{c}$, 
\[
\begin{array}{@{}ll@{\qquad}ll@{}}
\text{(modal)} & R_{M}^{c}(w)\subseteq R_{L}^{c}(w) & \text{(id)} & f^{c}(w,\varphi)\subseteq[[\varphi]]\\
\text{(soc-m)} & \text{If }w\in[[\Diamond_{M}\varphi]],\ \text{then }f^{c}(w,\varphi)\subseteq R_{M}^{c}(w) & \text{(sic)} & \text{If }w\in[[\Diamond_{L}\varphi]],\ \text{then }f^{c}(w,\varphi)\subseteq N^{c}
\end{array}
\]
\end{lemma}

\begin{proof}
(modal): Derived directly from the axiom $\Box_{L}\varphi\to\Box_{M}\varphi$.

(id): Derived directly from the axiom $\varphi>\varphi$.

(soc-m): Suppose $w\in[[\Diamond_{M}\varphi]]$. By the Truth Lemma,
$\Diamond_{M}\varphi\in w$. Assume there exists $u\in f^{c}(w,\varphi)$
such that $u\notin R_{M}^{c}(w)$. Then there exists $\psi$ such
that $\Box_{M}\psi\in w$ but $\psi\notin u$. By SOC-M, $\varphi>\psi\in w$.
Thus $\psi\in u$ by the definition of $f^{c}$, a contradiction.

(sic): Suppose $w\in[[\Diamond_{L}\varphi]]$. By the Truth Lemma,
$\Diamond_{L}\varphi\in w$. From the definition of $f^{c}$ in the
canonical model, $f^{c}(w,\varphi)\subseteq N^{c}$. 
\end{proof}

\begin{theorem}[Completeness] If $\Gamma\vDash_{SCP}\varphi$, then
$\Gamma\vdash_{SCP}\varphi$. \end{theorem}


\subsection{Decidability}

To construct a finite model, we first define a finite set of formulas
$\Sigma_{\varphi}$ that contains the target formula $\varphi$ and
all its relevant inferential components.

\begin{definition}[Subformula Closure] For a formula $\varphi$,
the closure $\Sigma_{\varphi}$ is defined as the smallest set of
formulas satisfying the following conditions: \label{def:subformula-closure} 
\begin{itemize}
\item $\varphi\in\Sigma_{\varphi}$; 
\item If $\psi\in\Sigma_{\varphi}$ and $\chi$ is a subformula of $\psi$,
then $\chi\in\Sigma_{\varphi}$; 
\item If $\psi\in\Sigma_{\varphi}$ and $\psi$ does not start with $\neg$,
then $\neg\psi\in\Sigma_{\varphi}$; 
\item If $\Box_{M}\psi\in\Sigma_{\varphi}$ or $\Diamond_{M}\psi\in\Sigma_{\varphi}$,
then $\psi\in\Sigma_{\varphi}$; 
\item If $\Box_{L}\psi\in\Sigma_{\varphi}$ or $\Diamond_{L}\psi\in\Sigma_{\varphi}$,
then $\psi,\Box_{M}\Box_{L}\psi\in\Sigma_{\varphi}$; 
\item If $\psi>\chi\in\Sigma_{\varphi}$, then $\psi,\chi,\Diamond_{L}\psi,\Diamond_{M}\psi,\Box_{L}\chi,\Box_{M}\chi,\psi>\psi\in\Sigma_{\varphi}$. 
\end{itemize}
\end{definition}

The inclusion of $\Box_{M}\Box_{L}\psi$ captures the nested interaction
between $R_{M}^{*}$ and $R_{L}^{*}$. While not a direct subformula
of $\Box_{L}\psi$, this single-layer expansion avoids infinite recursion.
Similarly, adding modal derivatives such as $\Diamond_{L}\psi$ and
$\Box_{M}\chi$ provides the necessary information for the selection
function to satisfy the SOC-M and SIC constraints. Crucially, since
each subformula in $\Sigma_{\varphi}$ generates only a fixed, finite
number of such derivatives, the closure set remains finite.

\begin{definition}[Equivalence Relation] Define an equivalence
relation $\sim$ between worlds in the canonical model $M^{c}$: \label{def:equivalence-relation}
\[
w\sim v\iff\forall\psi\in\Sigma_{\varphi},(\psi\in w\iff\psi\in v)
\]
\end{definition}

Let $\overline{w}=\{v\in W^{c}\mid v\sim w\}$ denote the equivalence
class of $w$. Since $\Sigma_{\varphi}$ is finite, the quotient set
$W^{*}=W^{c}/\sim$ is also finite.

\begin{definition}[Filtration Model] Given the canonical model
$M^{c}=(W^{c},N^{c},R_{M}^{c},R_{L}^{c},f^{c},V^{c})$, the filtration
model for $\Sigma_{\varphi}$ is defined as $M^{*}=(W^{*},N^{*},R_{M}^{*},R_{L}^{*},f^{*},V^{*})$,
where: 
\begin{itemize}
\item $W^{*}=\{\overline{w}\mid w\in W^{c}\}$; 
\item $N^{*}=\{\overline{w}\mid w\in N^{c}\}$; 
\item For any $\overline{w},\overline{v}\in N^{*}$, $\overline{w}R_{M}^{*}\overline{v}\iff\forall\Box_{M}\psi\in\Sigma_{\varphi},(\Box_{M}\psi\in w\iff\Box_{M}\psi\in v)$; 
\item For any $\overline{w},\overline{v}\in N^{*}$, $\overline{w}R_{L}^{*}\overline{v}\iff\forall\Box_{L}\psi\in\Sigma_{\varphi},(\Box_{L}\psi\in w\iff\Box_{L}\psi\in v)$; 
\item For any $\psi>\chi\in\Sigma_{\varphi}$, let $U_{\psi,w}=\{\overline{v}\mid\forall(\psi>\chi)\in\Sigma_{\varphi},(\psi>\chi\in w\Rightarrow\chi\in v)\}$.
Let 
\[
f^{*}(\overline{w},\psi)=\begin{cases}
U_{\psi,w}\cap R_{M}^{*}(\overline{w}) & \text{if }\Diamond_{M}\psi\in w\\
U_{\psi,w}\cap N^{*} & \text{if }\Diamond_{M}\psi\notin w\text{ and }\Diamond_{L}\psi\in w\\
U_{\psi,w} & \text{if }\Diamond_{L}\psi\notin w
\end{cases}
\]
\item For other $\psi$, let $f^{*}(\overline{w},\psi)=\varnothing$.
\item $V^{*}(\psi)=\{\overline{w}\mid w\in V^{c}(\psi)\}$. 
\end{itemize}
\end{definition}

Let $[[\varphi]]^{*}:=\{w\in W^{*}\mid M^{*},w\Vdash\varphi\}$.

\begin{lemma} The model is well-defined; that is, for any $\overline{w},\overline{v}\in N^{*}$
and any $\psi>\chi\in\Sigma_{\varphi}$: \label{lem:welldefinescp} 
\[
\begin{array}{@{}ll@{}}
1. \text{If } w\sim v,\text{ then } R_{M}^{*}(\overline{w})=R_{M}^{*}(\overline{v}); & 
2. R_{M}^{*} \text{ is an equivalence relation on } N^{*}; \\[4pt]
3. \text{If } w\sim v,\text{ then } R_{L}^{*}(\overline{w})=R_{L}^{*}(\overline{v}); & 
4. R_{L}^{*} \text{ is an equivalence relation on } N^{*}; \\[4pt]
5. \text{If } w\sim v,\text{ then } f^{*}(\overline{w},\psi)=f^{*}(\overline{v},\psi). & 
\end{array}
\]
\end{lemma}

\begin{proof}
Suppose $w\sim v$. For any $\overline{u}\in R_{M}^{*}(\overline{w})$,
we have $\forall\Box_{M}\psi\in\Sigma_{\varphi},(\Box_{M}\psi\in w\Rightarrow\Box_{M}\psi\in u)$.
Since $w\sim v$, we have $\forall\psi\in\Sigma_{\varphi},(\psi\in w\iff\psi\in v)$,
which implies $\forall\Box_{M}\psi\in\Sigma_{\varphi},(\Box_{M}\psi\in v\Rightarrow\Box_{M}\psi\in u)$.
Thus, $\overline{u}\in R_{M}^{*}(\overline{v})$, and $R_{M}^{*}(\overline{w})\subseteq R_{M}^{*}(\overline{v})$.
The reverse inclusion $R_{M}^{*}(\overline{v})\subseteq R_{M}^{*}(\overline{w})$
follows similarly.

It is easy to verify that $R_{M}^{*}$ is reflexive, transitive, and
symmetric. The case for $R_{L}^{*}$ is analogous.

Suppose $w\sim v$. For any $\overline{u}\in f^{*}(\overline{w},\psi)$,
we have $\forall(\psi>\chi)\in\Sigma_{\varphi},(\psi>\chi\in w\Rightarrow\chi\in u)$.
Since $w\sim v$, we have $\forall\psi\in\Sigma_{\varphi},(\psi\in w\iff\psi\in v)$,
which implies $\forall(\psi>\chi)\in\Sigma_{\varphi},(\psi>\chi\in v\Rightarrow\chi\in u)$.
Given $\Diamond_{M}\psi,\Diamond_{L}\psi\in\Sigma_{\varphi}$, we
have $\overline{u}\in f^{*}(\overline{v},\psi)$, so $f^{*}(\overline{w},\psi)\subseteq f^{*}(\overline{v},\psi)$.
The reverse inclusion follows similarly.
\end{proof}

\begin{lemma}[Filtration Lemma] For all $\psi\in\Sigma_{\varphi}$
and all $w\in W^{c}$: \label{lem:filtration-lemma} 
\[
M^{c},w\Vdash\psi\iff M^{*},\overline{w}\Vdash\psi
\]
\end{lemma}

\begin{proof}
For any $\overline{w}\notin N^{*}$ and $\psi\in\Sigma_{\varphi}$,
the truth condition yields $M^{c},w\Vdash\psi\iff w\in V^{c}(\psi)\iff\overline{w}\in V^{*}(\psi)\iff M^{*},\overline{w}\Vdash\psi$.

For any $\overline{w}\in N^{*}$, we proceed by induction on the complexity
of $\psi$. The cases for atomic propositions and Boolean connectives
are straightforward.

Case $\psi=\Box_{M}\gamma$:

$(\Rightarrow)$ Suppose $M^{c},w\Vdash\Box_{M}\gamma$. For any $\overline{v}$
such that $\overline{w}R_{M}^{*}\overline{v}$, by the definition
of $R_{M}^{*}$, $\Box_{M}\gamma\in w\iff\Box_{M}\gamma\in v$. Thus
$\Box_{M}\gamma\in v$. Since $M^{c}$ is an S5 model, $\gamma\in v$.
By IH, $M^{*},\overline{v}\Vdash\gamma$. Therefore
$M^{*},\overline{w}\Vdash\Box_{M}\gamma$.

$(\Leftarrow)$ Suppose $M^{c},w\nVdash\Box_{M}\gamma$. In $M^{c}$,
there exists $v\in N^{c}$ such that $wR_{M}^{c}v$ and $v\nVdash\gamma$.
We show that $\overline{w}R_{M}^{*}\overline{v}$ holds. For any $\Box_{M}\chi\in\Sigma_{\varphi}$,
if $\Box_{M}\chi\in w$, then $\chi\in v$ because $wR_{M}^{c}v$;
by S5 properties, $\Box_{M}\chi\in v$ also holds. The converse is
similar. Thus $\overline{w}R_{M}^{*}\overline{v}$ holds. Since $M^{c},v\nVdash\gamma$,
IH yields $M^{*},\overline{v}\nVdash\gamma$,
hence $M^{*},\overline{w}\nVdash\Box_{M}\gamma$.

Case $\psi=\Box_{L}\gamma$: Analogous to the $\Box_{M}\gamma$ case.

Case $\psi=\alpha>\beta$:

$(\Rightarrow)$ Suppose $M^{c},w\Vdash\alpha>\beta$. For any $\overline{v}\in f^{*}(\overline{w},\alpha)$,
by the definition of $f^{*}$, we have $\forall(\alpha>\chi)\in\Sigma_{\varphi},(\alpha>\chi\in w\Rightarrow\chi\in v)$.
Thus $\beta\in v$. By IH, $M^{*},\overline{v}\Vdash\beta$.
Therefore $f^{*}(\overline{w},\alpha)\subseteq[[\beta]]^{*}$, i.e.,
$M^{*},\overline{w}\Vdash\alpha>\beta$.

$(\Leftarrow)$ Suppose $M^{c},w\nVdash\alpha>\beta$. In $M^{c}$,
there exists a world $u$ such that $\{\chi\mid\alpha>\chi\in w\}\subseteq u$
and $\beta\notin u$. Consider the equivalence class $\overline{u}$. 

If $\Diamond_{M}\alpha\in w$: For any $\Box_{M}\gamma\in\Sigma_{\varphi}$,
if $\Box_{M}\gamma\in w$, then $\alpha>\gamma\in w$ (by SOC-M).
From the construction of $u$, $\gamma\in u$. Since $M^{c}$ is S5,
$\Box_{M}\gamma\in w\Rightarrow\Box_{M}\Box_{M}\gamma\in w\Rightarrow\alpha>\Box_{M}\gamma\in w\Rightarrow\Box_{M}\gamma\in u$.
The converse $\neg\Box_{M}\gamma\in w\Rightarrow\neg\Box_{M}\gamma\in u$
follows similarly. Thus $\overline{w}R_{M}^{*}\overline{u}$ holds
by the definition of $R_{M}^{*}$, so $\overline{u}\in f^{*}(\overline{w},\alpha)$.
By IH, $M^{*},\overline{u}\nVdash\beta$. Hence
$M^{*},\overline{w}\nVdash\alpha>\beta$. 

If $\Diamond_{M}\alpha\notin w$ and $\Diamond_{L}\alpha\in w$. By the
properties of the selection function in the canonical model, we have
$u\in N^{c}$, which implies $\overline{u}\in N^{*}$. It follows
directly that $\overline{u}\in f^{*}(\overline{w},\alpha)$. By the
induction hypothesis, $M^{*},\overline{u}\nVdash\beta$, hence $M^{*},\overline{w}\nVdash\alpha>\beta$. 

If $\Diamond_{L}\alpha\notin w$: Clearly $\overline{u}\in f^{*}(\overline{w},\alpha)$.
By IH, $M^{*},\overline{u}\nVdash\beta$. Hence
$M^{*},\overline{w}\nVdash\alpha>\beta$. 
\end{proof}

\begin{lemma} The filtration model satisfies the semantic constraints:
for all $\psi>\chi\in\Sigma_{\varphi}$ and $\overline{w}\in N^{*}$,
\label{lem:filter-constraints} 
\[
\begin{array}{@{}ll@{\qquad}ll@{}}
\text{(modal)} & R_{M}^{*}(\overline{w})\subseteq R_{L}^{*}(\overline{w}) & \text{(id)} & f^{*}(\overline{w},\psi)\subseteq[[\psi]]^{*}\\
\text{(soc-m)} & \text{If }M^{*},\overline{w}\Vdash\Diamond_{M}\psi,\ \text{then }f^{*}(\overline{w},\psi)\subseteq R_{M}^{*}(\overline{w}) & \text{(sic)} & \text{If }M^{*},\overline{w}\Vdash\Diamond_{L}\psi,\ \text{then }f^{*}(\overline{w},\psi)\subseteq N^{*}
\end{array}
\]
\end{lemma}

\begin{proof}
(modal): For any $\overline{v}\in R_{M}^{*}(\overline{w})$, we show
$\overline{w}R_{L}^{*}\overline{v}$, i.e., for all $\Box_{L}\psi\in\Sigma_{\varphi}$,
$\Box_{L}\psi\in w\iff\Box_{L}\psi\in v$. Suppose $\Box_{L}\psi\in w$.
Since $\Box_{L}\psi\in\Sigma_{\varphi}$, we have $\Box_{M}\Box_{L}\psi\in\Sigma_{\varphi}$
by Definition \ref{def:subformula-closure}. Since $\overline{w}\in N^{*}$,
$w$ contains at least one MCS $w'\in N^{c}$. By Proposition \ref{prop:necessity-LM}.3,
$\vDash\Box_{L}\psi\to\Box_{M}\Box_{L}\psi$, so $\Box_{M}\Box_{L}\psi\in w'$.
By the definition of $R_{M}^{*}$, $\Box_{M}\Box_{L}\psi\in v$. By
axiom TM and $v\sim v'$, we have $\Box_{L}\psi\in v'$ and thus $\Box_{L}\psi\in v$.
For the other direction, suppose $\Box_{L}\psi\in v$. Let $v'\in\overline{v}\cap N^{c}$.
By Axiom $\Box_{L}\psi\to\Box_{M}\Box_{L}\psi$, we have $\Box_{M}\Box_{L}\psi\in v'$.
Thus $\Box_{M}\Box_{L}\psi\in v$. As $R_{M}^{*}$ is symmetric, we
have $\Box_{M}\Box_{L}\psi\in w$. Let $w'\in\overline{w}\cap N^{c}$.
By Axiom TM, we have $\Box_{L}\psi\in w'$. Thus $\Box_{L}\psi\in w$.

(id): Let $\overline{w}\in N^{*}$. Let $v'\in\overline{w}\cap N^{c}$.
By Axiom ID, $\psi>\psi\in w'$. As $w\sim w'$, we also have $\psi>\psi\in w$.
For any $\overline{v}\in U_{\psi,w}$, by the definition of $U_{\psi,w}$,
it follows that $\psi\in v$. By the filtration lemma, $M^{*},\overline{v}\Vdash\psi$,
i.e., $f^{*}(\overline{w},\psi)\subseteq[[\psi]]^{*}$.

(soc-m): Suppose $M^{*},\overline{w}\Vdash\Diamond_{M}\psi$. By Lemma
\ref{lem:filtration-lemma}, $\Diamond_{M}\psi\in w$. By definition,
$f^{*}(\overline{w},\psi)\subseteq R_{M}^{*}(\overline{w})$.

(sic): Suppose $M^{*},\overline{w}\Vdash\Diamond_{L}\psi$. By Lemma
\ref{lem:filtration-lemma}, $\Diamond_{L}\psi\in w$. By definition,
$f^{*}(\overline{w},\psi)\subseteq N^{*}$.
\end{proof}

\begin{theorem} $\mathbf{SCP}$ has finite model property and
is decidable. \end{theorem}

\section{Non-vacuous Semantics \texorpdfstring{($\mathbf{SCP1}$)}{(SCP1)}}

\label{sec:scp1}

The preceding section established a general semantic framework for $\mathbf{SCP}$. While flexible, $\mathbf{SCP}$ exhibits two theoretical vulnerabilities regarding non-vacuism and logical preservation.

First, $\mathbf{SCP}$ permits the selection function $f$ to be empty. For an impossible antecedent $\varphi$, a model can simply assign $f(w, \varphi) = \varnothing$, making $\varphi > \psi$ vacuously true. If impossible worlds are rich enough, this is counterintuitive and renders the non-vacuist agenda incompletely realized.

Second, the logical counterpart to (SOC-M), namely (SOC-L): $\Diamond_L \varphi \land \Box_L \psi \to (\varphi > \psi)$, is invalid in the base system. Consider a model where $N=\{w, v\}$ with disjoint logical clusters $R_L(w)=\{w\}$ and $R_L(v)=\{v\}$. Let $w \in V(\varphi) \cap V(\psi)$ and $v \in V(\varphi) \setminus V(\psi)$. Clearly, $w \Vdash \Diamond_L \varphi \land \Box_L \psi$. However, the semantic constraint (sic) only requires $f(w, \varphi) \subseteq N$. If the selection yields $f(w, \varphi) = \{v\}$, we get $w \not\Vdash \varphi > \psi$. This invalidity highlights a flaw in $\mathbf{SCP}$: the evaluation of conditionals can cross into disconnected logical clusters, causing established logical necessities ($\Box_L \psi$) to fail.

To provide a more rigorous characterization of counterpossible reasoning, we introduce the extended system $\mathbf{SCP1}$. This
system augments the semantics with two core constraints: \textbf{(ne)}
(non-vacuism) ensures $f$ is never empty by evaluating contradictions
within a sufficiently rich set of impossible worlds $I$; and \textbf{(rl)}
(universal accessibility) treats $R_{L}$ as a universal relation
over $N$, reflecting the global consistency of logical laws.

\subsection{Semantics}

\begin{definition}[$\mathbf{SCP1}$ Constraints] In addition to
Definition \ref{def:scp-constraints}, $\mathbf{SCP1}$ models must
satisfy: 
\begin{itemize}
\item[{(rl)}] $R_{L}(w)=N$ for all $w\in N$. 
\item[{(ne)}] $f(w,\varphi)\neq\varnothing$ for all $w\in N,\varphi\in\mathscr{L}$. 
\end{itemize}
\end{definition}
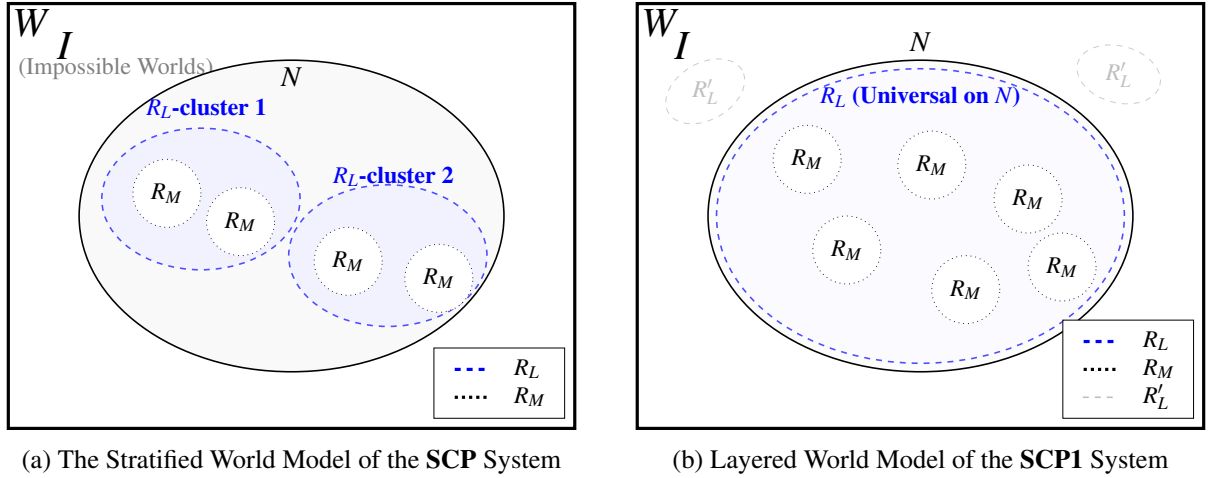
\begin{figure}[htbp]
    \centering
    \begin{subfigure}{0.48\textwidth}
        \centering
        \resizebox{\textwidth}{!}{
            \begin{tikzpicture}[
                outerworld/.style={rectangle, draw, minimum width=10cm, minimum height=7.5cm, ultra thick},
                normalzone/.style={ellipse, draw, fill=gray!5, minimum width=7.5cm, minimum height=5.5cm, thick},
                rlcluster/.style={ellipse, draw, dashed, blue!70, fill=blue!5, minimum width=3.5cm, minimum height=2.5cm, thick},
                rmclass/.style={circle, draw, dotted, black, fill=white, minimum size=1.2cm, inner sep=0pt}
            ]
                \node[outerworld] (W) at (0,0) {};
                \node[anchor=north west, font=\huge\bfseries] at (W.north west) {$W$};

                \node[font=\huge] at (-4, 3) {$I$};
                \node[font=\large, text=gray] at (-3.1, 2.6) {(Impossible Worlds)};

                \node[normalzone] (N) at (0,0) {};
                \node[anchor=north, font=\Large\bfseries] at (N.north) {$N$};

                \node[rlcluster] (RL1) at (-1.6, 0.3) {};
                \node[rlcluster] (RL2) at (1.7, -0.7) {};
                \node[blue, font=\large\bfseries] at (-1.5, 1.9) {$R_L$-cluster 1};
                \node[blue, font=\large\bfseries] at (1.8, 0.7) {$R_L$-cluster 2};

                \node[rmclass] at (-2.2, 0.4) {\large $R_M$};

                \node[rmclass] at (-0.9, -0.1) {\large $R_M$};
                \node[rmclass] at (1.0, -0.8) {\large $R_M$};
 
                \node[rmclass] at (2.6, -1.1) {\large $R_M$};

                \node[draw, fill=white, shift={(-0.2,0.2)}, anchor=south east] at (W.south east) {
                    \large
                    \begin{tabular}{ll}
                    \tikz[baseline=-0.5ex]\draw[dashed, blue, ultra thick] (0,0) -- (0.6,0); &  $R_L$ \\
                    \tikz[baseline=-0.5ex]\draw[dotted, black, ultra thick] (0,0) -- (0.6,0); &  $R_M$ \\
                    \end{tabular}
                };
            \end{tikzpicture}
        }
        \caption{The Stratified World Model of the $\mathbf{SCP}$ System}
        \label{fig:scpmodel}
    \end{subfigure}
    \hfill
    \begin{subfigure}{0.48\textwidth}
        \centering
        \resizebox{\textwidth}{!}{
             \begin{tikzpicture}[
                outerworld/.style={rectangle, draw, minimum width=10cm, minimum height=7.5cm, ultra thick},
                normalzone/.style={ellipse, draw, fill=gray!2, minimum width=7.5cm, minimum height=5.5cm, thick},
                rluniversal/.style={ellipse, draw, dashed, blue!70, fill=blue!2, minimum width=7.2cm, minimum height=5.2cm, thick},
                rmclass/.style={circle, draw, dotted, black, fill=white, minimum size=1.2cm, inner sep=0pt},
                rlprime/.style={ellipse, draw, dashed, gray!40, minimum width=1.5cm, minimum height=1cm, thin}
            ]
                \node[outerworld] (W) at (0,0) {};
                \node[anchor=north west, font=\huge\bfseries] at (W.north west) {$W$};

                \node[rlprime, rotate=30] at (-3.8, 2.2) {};
                \node[rlprime, rotate=-15] at (3.5, 2.5) {};
                \node[gray!50, font=\large] at (-3.8, 2.2) {$R'_L$};
                \node[gray!50, font=\large] at (3.5, 2.5) {$R'_L$};
                \node[font=\huge] at (-4.2, 3) {$I$};

                \node[normalzone] (N) at (0,0) {};
                \node[anchor=south, font=\Large\bfseries] at (N.north) {$N$};

                \node[rluniversal] (RL) at (0,0) {};
                \node[blue, font=\large\bfseries] at (0, 2.1) {$R_L$ (Universal on $N$)};

                \node[rmclass] at (-2.0, 1.0) {\large $R_M$};
                \node[rmclass] at (-1.3, -0.6) {\large $R_M$};
                \node[rmclass] at (0.2, 0.9) {\large $R_M$};
                \node[rmclass] at (1.9, 0.3) {\large $R_M$};
                \node[rmclass] at (0.8, -1.3) {\large $R_M$};
                \node[rmclass] at (2.5, -0.9) {\large $R_M$};

                \node[draw, fill=white, shift={(-0.2,0.2)}, anchor=south east] at (W.south east) {
                    \large
                    \begin{tabular}{ll}
                    \tikz[baseline=-0.5ex]\draw[dashed, blue, ultra thick] (0,0) -- (0.6,0); &  $R_L$ \\
                    \tikz[baseline=-0.5ex]\draw[dotted, black, ultra thick] (0,0) -- (0.6,0); &  $R_M$ \\
                    \tikz[baseline=-0.5ex]\draw[dashed, gray!50, thick] (0,0) -- (0.6,0); &  $R'_L$ \\
                    \end{tabular}
                };
            \end{tikzpicture}
        }
        \caption{Layered World Model of the $\mathbf{SCP1}$ System}
    \label{fig:scp1model}
    \end{subfigure}
    \caption{Stratified World Models for $\mathbf{SCP}$ and $\mathbf{SCP1}$}
\end{figure}
Figure \ref{fig:scp1model} illustrates the distinctions of $\mathbf{SCP1}$: 

\textbf{Universal Logical Accessibility ($R_{L}$)}: $R_{L}$ is defined as a universal relation over $N$. Since all normal worlds share the same classical laws, they are mutually accessible at the logical level, with metaphysical equivalence classes ($R_{M}$) nested within this universal domain.
    
\textbf{Evaluation of Contradictions}: The $I$ region contains clusters
($R'_{L}$) that obey alternative logical rules.\footnote{See \cite{bjerring2014Counterpossibles} on using non-classical closure
sets as impossible worlds.} Although these regions
are not the focus of the system's deduction, their existence provides
a rational explanation for why the \textbf{(ne)} constraint can provide
concrete evaluation worlds for inconsistent hypotheses. 

Let $\mathsf{SCP1}$ be the class of models satisfying
these constraints; $\vDash_{\mathsf{SCP1}}$ is hereafter abbreviated
as $\vDash$.

\subsection{Soundness and completeness}

The axiomatic system for $\mathbf{SCP1}$ is an extension of the base
system $\mathbf{SCP}$, obtained by adding the following two axioms:
\begin{align*}
 & \text{(NE)} &  & \Diamond_{L}\varphi\to\neg(\varphi>\neg\varphi)\\
 & \text{(SOC-L)} &  & \Diamond_{L}\varphi\land\Box_{L}\psi\to(\varphi>\psi)
\end{align*}

\begin{lemma} \label{lem:consistent-N} $[\varphi]\cap[\neg\varphi]\cap N=\varnothing$.
\end{lemma}

\begin{proof}
On the set of normal worlds $N$, truth values are determined by the
recursive definition in Definition \ref{def:scp-truth}. According
to the condition $M,w\Vdash\neg\varphi\iff M,w\nVdash\varphi$, it
follows that $[\varphi]\cap[\neg\varphi]\cap N=\varnothing$.
\end{proof}

\begin{theorem}[Soundness of $\mathbf{SCP1}$]\label{soundness-scp1}
For any $\Gamma\cup\{\varphi\}\subseteq\mathscr{L}$, if $\Gamma\vdash_{\mathbf{SCP1}}\varphi$,
then $\Gamma\vDash_{\mathbf{SCP1}}\varphi$. \end{theorem}

\begin{proof}
Refer to Theorem \ref{theorem:scp-soundness}. We only need to prove
the validity of (SOC-L) and (NE).

\textbf{(SOC-L)}: Suppose $w\in N$ and $M,w\Vdash\Diamond_{L}\varphi\land\Box_{L}\psi$
holds. Since $M,w\Vdash\Box_{L}\psi$, we have $N\subseteq[\psi]$.
Since $M,w\Vdash\Diamond_{L}\varphi$, it follows from the semantic
constraints that $f(w,\varphi)\subseteq N$. Thus, $f(w,\varphi)\subseteq[\psi]$,
which means $M,w\Vdash\varphi>\psi$.

\textbf{(NE)}: Suppose $w\in N$ and $M,w\Vdash\Diamond_{L}\varphi$
holds. This implies $f(w,\varphi)\subseteq N$. By constraint (ne),
$f(w,\varphi)\neq\varnothing$; thus, there exists $u\in f(w,\varphi)$.
By constraint (id), $f(w,\varphi)\subseteq[\varphi]$. According to
Lemma \ref{lem:consistent-N}, $f(w,\varphi)\nsubseteq[\neg\varphi]$.
Therefore, $M,w\nVdash\varphi>\neg\varphi$, which is equivalent to
$M,w\Vdash\neg(\varphi>\neg\varphi)$.
\end{proof}

We construct a generated submodel $M^{cg}$ based on the canonical
model $M^{c}$ defined in Definition \ref{def:canonical-model}, except
that MCSs are based on $\mathbf{SCP1}$.

\begin{definition} Let $\Gamma$ be an $\mathbf{SCP1}$-consistent
set, and let $w_{\Gamma}\in N^{c}$ be a maximal consistent set (MCS)
containing $\Gamma$. We construct the generated submodel $M^{cg}=(W^{cg},N^{cg},R_{M}^{cg},R_{L}^{cg},f^{cg},V^{cg})$
as follows: \end{definition}
\begin{itemize}
\item $W^{cg}=\mathcal{P}(\mathscr{L})$ 
\begin{itemize}
\item $N^{cg}=N^{c}\cap R_{L}^{c}(w_{\Gamma})$ 
\item $I^{cg}=W^{c}-N^{cg}$ 
\end{itemize}
\item $R_{M}^{cg}=R_{M}^{c}\cap(R_{L}^{c}(w_{\Gamma})\times R_{L}^{c}(w_{\Gamma}))$ 
\item $R_{L}^{cg}=R_{L}^{c}\cap(R_{L}^{c}(w_{\Gamma})\times R_{L}^{c}(w_{\Gamma}))$ 
\item $f^{cg}(w,\varphi)=\begin{cases}
\{x\in N^{cg}\mid\{\psi\mid\varphi>\psi\in w\}\subseteq x\} & \text{if }\Diamond_{L}\varphi\in w\\
\{x\in W^{cg}\mid\{\psi\mid\varphi>\psi\in w\}\subseteq x\} & \text{if }\Diamond_{L}\varphi\notin w
\end{cases}$ 
\item $V^{cg}(\varphi)=V^{c}(\varphi)$ 
\end{itemize}
Let $[[\varphi]]^{cg}:=\{w\in W^{cg}\mid M^{cg},w\Vdash\varphi\}$.

Through the construction of the generated submodel, we first verify
whether it satisfies the universal relation requirement for logical
accessibility in $\mathbf{SCP1}$.

\begin{lemma} $R_{L}^{cg}$ is a universal relation on $N^{cg}$.
\label{lem:universal-RL} \end{lemma}
\begin{proof}
Since $N^{cg}=N^{c}\cap R_{L}^{c}(w_{\Gamma})$ and $R_{L}^{c}(w_{\Gamma})\subseteq N^{c}$,
it follows that $R_{L}^{c}(w_{\Gamma})=N^{cg}$. Thus, $R_{L}^{cg}(w_{\Gamma})=R_{L}^{c}(w_{\Gamma})$.

We show that for all $x,y\in R_{L}^{cg}(w_{\Gamma})$, $xR_{L}^{cg}y$
holds. Suppose $\Box_{L}\varphi\in x$; we must show $\varphi\in y$.
By axiom 4L, $\Box_{L}\Box_{L}\varphi\in x$. Suppose for contradiction
that $\Box_{L}\Box_{L}\varphi\notin w_{\Gamma}$. Then $\Diamond_{L}\Diamond_{L}\neg\varphi\in w_{\Gamma}$.
By axiom 5L, $\Box_{L}\Diamond_{L}\Diamond_{L}\neg\varphi\in w_{\Gamma}$,
which implies $\Diamond_{L}\Diamond_{L}\neg\varphi\in x$, a contradiction.
Therefore, $\Box_{L}\Box_{L}\varphi\in w_{\Gamma}$. By axiom TL,
$\Box_{L}\varphi\in w_{\Gamma}$, which implies $\varphi\in y$ since
$w_{\Gamma}R_{L}^{c}y$.
\end{proof}

\begin{lemma}[Truth Lemma] \label{lem:truth-lemma-1} For all $w\in W^{cg}$,
$M^{cg},w\Vdash\alpha\iff\alpha\in w$. \end{lemma}

\begin{proof}

For any $w\in I^{cg}$ and any $\alpha\in\mathscr{L}$, the truth
conditions for anti-logical worlds give $M^{cg},w\Vdash\alpha\iff w\in V^{cg}(\alpha)\iff\alpha\in w$.

For any $w\in N^{cg}$, we proceed by induction on the complexity
of $\alpha$. The cases for atomic propositions and Boolean connectives
are standard.

Case $\alpha=\Box_{M}\varphi$:

$(\Leftarrow)$ Suppose $\Box_{M}\varphi\in w$. For all $x\in R_{M}^{cg}(w)$,
since $R_{M}^{cg}(w)\subseteq R_{M}^{c}(w)$, we have $\varphi\in x$.
By IH, $M^{cg},x\Vdash\varphi$. Thus $M^{cg},w\Vdash\Box_{M}\varphi$.

$(\Rightarrow)$ Suppose $\Box_{M}\varphi\notin w$. In the canonical
model, there exists an MCS $u\in N^{c}$ such that $wR_{M}^{c}u$
and $\neg\varphi\in u$. We must strictly show $u\in N^{cg}$. Since
$w\in N^{cg}=R_{L}^{c}(w_{\Gamma})$, we have $w_{\Gamma}R_{L}^{c}w$.
By the axiom $\Box_{L}\psi\to\Box_{M}\psi$, $R_{M}^{c}\subseteq R_{L}^{c}$,
so $wR_{L}^{c}u$. S5 logic guarantees $R_{L}^{c}$ is transitive,
hence $w_{\Gamma}R_{L}^{c}u$. This implies $u\in R_{L}^{c}(w_{\Gamma})$,
so $u\in N^{cg}$. Since $w,u\in N^{cg}$ and $wR_{M}^{c}u$, we have
$wR_{M}^{cg}u$. By IH, $M^{cg},u\nVdash\varphi$,
therefore $M^{cg},w\nVdash\Box_{M}\varphi$.

Case $\alpha=\Box_{L}\varphi$:

$(\Leftarrow)$ Suppose $\Box_{L}\varphi\in w$. By Lemma \ref{lem:universal-RL},
$R_{L}^{cg}$ is a universal relation on $N^{cg}$. For any $x\in N^{cg}$,
we have $wR_{L}^{cg}x$, which implies $\varphi\in x$. By the induction
hypothesis, $M^{cg},x\Vdash\varphi$, thus $M^{cg},w\Vdash\Box_{L}\varphi$.

$(\Rightarrow)$ Suppose $\Box_{L}\varphi\notin w$. There exists
an MCS $u\in N^{c}$ such that $wR_{L}^{c}u$ and $\neg\varphi\in u$.
Since $w\in N^{cg}$, $w_{\Gamma}R_{L}^{c}w$. By the transitivity
of $R_{L}^{c}$, $w_{\Gamma}R_{L}^{c}u$, which perfectly places $u\in N^{cg}$.
Thus $wR_{L}^{cg}u$. By IH, $M^{cg},u\nVdash\varphi$,
hence $M^{cg},w\nVdash\Box_{L}\varphi$.

Case $\alpha=\varphi>\psi$:

$(\Leftarrow)$ Suppose $\varphi>\psi\in w$. By the definition of
$f^{cg}$, for any $x\in f^{cg}(w,\varphi)$, $\psi\in x$. By the
induction hypothesis, $M^{cg},x\Vdash\psi$. Thus $M^{cg},w\Vdash\varphi>\psi$.

$(\Rightarrow)$ Suppose $\varphi>\psi\notin w$. 

If $\Diamond_{L}\varphi\notin w$: Construct $u=\{\chi\mid\varphi>\chi\in w\}$.
Since $\varphi>\psi\notin w$, $\psi\notin u$. By the definition
of $f^{cg}$, $u\in f^{cg}(w,\varphi)$. By IH,
$M^{cg},u\nVdash\psi$, hence $M^{cg},w\nVdash\varphi>\psi$. 

If $\Diamond_{L}\varphi\in w$: We must construct a world $u\in N^{cg}$
such that $u\in f^{cg}(w,\varphi)$ and $\psi\notin u$. We first
verify that the set $\Gamma'=\{\chi\mid\varphi>\chi\in w\}\cup\{\neg\psi\}$
is consistent. If not, $\vdash(\chi_{1}\land\dots\land\chi_{n})\to\psi$.
By the (SIC) rule, $\vdash\Diamond_{L}\varphi\to(((\varphi>\chi_{1})\land\dots\land(\varphi>\chi_{n}))\to(\varphi>\psi))$.
Since $\Diamond_{L}\varphi\in w$ and $\varphi>\chi_{i}\in w$, this
forces $\varphi>\psi\in w$, yielding a contradiction.

By Lindenbaum's Lemma, $\Gamma'$ can be extended to an MCS $u\in N^{c}$.
To prove $u\in N^{cg}$, we must show $u\in R_{L}^{c}(w_{\Gamma})$.
For any $\Box_{L}\gamma\in w_{\Gamma}$, by Axiom 4L, $\Box_{L}\Box_{L}\gamma\in w_{\Gamma}$.
Since $w\in R_{L}^{c}(w_{\Gamma})$, we have $\Box_{L}\gamma\in w$.
By the (SOC-L) axiom, $\Diamond_{L}\varphi\land\Box_{L}\gamma\to(\varphi>\gamma)$.
Since $\Diamond_{L}\varphi\in w$, it follows that $\varphi>\gamma\in w$,
meaning $\gamma\in\Gamma'\subseteq u$. Thus $w_{\Gamma}R_{L}^{c}u$,
which confirms $u\in N^{cg}$. Since $\{\chi\mid\varphi>\chi\in w\}\subseteq u$,
we have $u\in f^{cg}(w,\varphi)$. Because $\neg\psi\in u$, $\psi\notin u$.
By IH, $M^{cg},u\nVdash\psi$. Therefore, $M^{cg},w\nVdash\varphi>\psi$.
\end{proof}

Finally, we verify that $M^{cg}$ satisfies the specific constraints
of $\mathbf{SCP1}$.

\begin{lemma}\label{lem:cg-constraints-scp1} The canonical model
$M^{cg}$ satisfies the following model constraints for all $w\in N^{cg}$:
\[
\begin{array}{@{}ll@{\qquad}ll@{}}
\text{(rl)} & R_{L}^{cg}(w)=N^{cg} & \text{(id)} & f^{cg}(w,\varphi)\subseteq[[\varphi]]^{cg}\\
\text{(ne)} & f^{cg}(w,\varphi)\neq\varnothing & \text{(soc-m)} & \text{If }w\in[[\Diamond_{M}\varphi]],\ \text{then }f^{cg}(w,\varphi)\subseteq R_{M}^{cg}(w)\\
\text{(modal)} & R_{M}^{cg}(w)\subseteq R_{L}^{cg}(w) & \text{(sic)} & \text{If }w\in[[\Diamond_{L}\varphi]],\ \text{then }f^{cg}(w,\varphi)\subseteq N^{cg}
\end{array}
\]
\end{lemma}

\begin{proof}







\textbf{(rl)}: Proven in Lemma \ref{lem:universal-RL}.

\textbf{(ne)}: If $\Diamond_{L}\varphi\notin w$, let $u=\{\psi\mid\varphi>\psi\in w\}$.
Since $W^{cg}=\wp(\mathcal{L})$, $u$ is a valid world in $W^{cg}$,
ensuring $f^{cg}(w,\varphi)\neq\varnothing$. If $\Diamond_{L}\varphi\in w$,
by axiom (NE), $\neg(\varphi>\neg\varphi)\in w$. By the Truth Lemma
on $M^{cg}$, $M^{cg},w\nVdash\varphi>\neg\varphi$, which strictly
implies $f^{cg}(w,\varphi)\nsubseteq[[\neg\varphi]]^{cg}$. Therefore,
$f^{cg}(w,\varphi)$ cannot be empty.

\textbf{(modal)}: By definition, $R_{M}^{cg}=R_{M}^{c}\cap(R_{L}^{c}(w_{\Gamma})\times R_{L}^{c}(w_{\Gamma}))$
and $R_{L}^{cg}=R_{L}^{c}\cap(R_{L}^{c}(w_{\Gamma})\times R_{L}^{c}(w_{\Gamma}))$.
Since $R_{M}^{c}(x)\subseteq R_{L}^{c}(x)$ holds for all $x\in N^{c}$
in the canonical model, it algebraically follows that $R_{M}^{cg}(w)\subseteq R_{L}^{cg}(w)$
for any $w\in N^{cg}$.

\textbf{(id)}: For any $u\in f^{cg}(w,\varphi)$, by the axiom (ID),
$\varphi>\varphi\in w$. By the definition of $f^{cg}$, this requires
$\varphi\in u$. By the Truth Lemma on $M^{cg}$, $u\in[[\varphi]]^{cg}$.
Thus $f^{cg}(w,\varphi)\subseteq[[\varphi]]^{cg}$.

\textbf{(soc-m)}: Suppose $w\in[[\Diamond_{M}\varphi]]^{cg}$. By
the Truth Lemma, $\Diamond_{M}\varphi\in w$. Since $\vdash\Diamond_{M}\varphi\to\Diamond_{L}\varphi$,
we have $\Diamond_{L}\varphi\in w$. Therefore, by the definition
of $f^{cg}$, $f^{cg}(w,\varphi)\subseteq N^{cg}$. Now, assume for
contradiction there exists $u\in f^{cg}(w,\varphi)$ such that $u\notin R_{M}^{cg}(w)$.
Since $u\in N^{cg}$, it must be that $u\notin R_{M}^{c}(w)$. By
the definition of $R_{M}^{c}$, there exists $\psi$ such that $\Box_{M}\psi\in w$
but $\psi\notin u$. By the (SOC-M) axiom, $\varphi>\psi\in w$, which
requires $\psi\in u$ according to the definition of $f^{cg}$, establishing
a contradiction. Thus $f^{cg}(w,\varphi)\subseteq R_{M}^{cg}(w)$.

\textbf{(sic)}: Suppose $w\in[[\Diamond_{L}\varphi]]^{cg}$. By the
Truth Lemma, $\Diamond_{L}\varphi\in w$. By the precise definition
of the selection function in $M^{cg}$ for this condition, $f^{cg}(w,\varphi)\subseteq N^{cg}$.
\end{proof}

\begin{theorem}[Completeness of $\mathbf{SCP1}$] If $\Gamma\vDash_{\mathbf{SCP1}}\varphi$,
then $\Gamma\vdash_{\mathbf{SCP1}}\varphi$. \end{theorem}

\subsection{Decidability}

We prove that the system remains decidable despite the introduction
of the non-emptiness constraint and the universal logical accessibility
relation.

The subformula closure $\Sigma_{\varphi}$ for $\mathbf{SCP1}$ follows
Definitions \ref{def:subformula-closure} and \ref{def:equivalence-relation},
except that the expansion $\Box_{M}\Box_{L}\psi$ is no longer required
for the $\mathbf{SCP1}$ closure set.

The filtration model for $\mathbf{SCP1}$ is constructed from a generated
submodel of the canonical model. While the overall approach is consistent
with standard methods, the definition of $R_{L}^{*}$ is modified
to accommodate the specific constraints of $\mathbf{SCP1}$.

\begin{definition}[Filtration Model] Given a generated submodel
$M^{cg}=(W^{cg},N^{cg},R_{M}^{cg},R_{L}^{cg},f^{cg},V^{cg})$, the
filtration model targeting $\Sigma_{\varphi}$ is defined as $M^{*}=(W^{*},N^{*},R_{M}^{*},R_{L}^{*},f^{*},V^{*})$,
where: \end{definition}
\begin{itemize}
\item $W^{*}=\{\overline{w}\mid w\in W^{cg}\}$; 
\item $N^{*}=\{\overline{w}\mid w\in N^{cg}\}$; 
\item For any $\overline{w},\overline{v}\in N^{*}$, $\overline{w}R_{M}^{*}\overline{v}\iff\forall\Box_{M}\psi\in\Sigma_{\varphi},(\Box_{M}\psi\in w\iff\Box_{M}\psi\in v)$; 
\item For any $\overline{w},\overline{v}\in N^{*}$, $\overline{w}R_{L}^{*}\overline{v}\iff\exists w'\in\overline{w},\exists v'\in\overline{v}$
such that $w'R_{L}^{cg}v'$; 
\item For any $\psi>\chi\in\Sigma_{\varphi}$, let $U_{\psi,w}=\{\overline{v}\mid\forall(\psi>\chi)\in\Sigma_{\varphi},(\psi>\chi\in w\Rightarrow\chi\in v)\}$.
Let 
\[
f^{*}(\overline{w},\psi)=\begin{cases}
U_{\psi,w}\cap R_{M}^{*}(\overline{w}) & \text{if }\Diamond_{M}\psi\in w\\
U_{\psi,w}\cap N^{*} & \text{if }\Diamond_{M}\psi\notin w\text{ and }\Diamond_{L}\psi\in w\\
U_{\psi,w}\cap W^{*} & \text{if }\Diamond_{L}\psi\notin w
\end{cases}
\]
\item For other $\psi$, let
\[
f^{*}(\overline{w},\psi)=\begin{cases}
[[\psi]]^{*}\cap R_{M}^{*}(\overline{w}) & \text{if }M^{*},\overline{w}\Vdash\Diamond_{M}\psi\\{}
[[\psi]]^{*}\cap N^{*} & \text{if }M^{*},\overline{w}\nVdash\Diamond_{M}\psi\text{ and }M^{*},\overline{w}\Vdash\Diamond_{L}\psi\\{}
[[\psi]]^{*} & \text{otherwise}
\end{cases},
\]
where $[[\psi]]^{*}:=\{w\in W^{*}\mid M^{*},w\Vdash\psi\}$;
\item $V^{*}(\psi)=\{\overline{w}\mid w\in V^{cg}(\psi)\}$. 
\end{itemize}

\begin{lemma} \label{lem:welldefinescp1} The model is well-defined.
For any $\overline{w},\overline{v}\in N^{*}$ and any $(\psi>\chi)\in\Sigma_{\varphi}$: 
\[
\begin{array}{@{}ll@{}}
1. \text{If } w\sim v,\text{ then } R_{M}^{*}(\overline{w})=R_{M}^{*}(\overline{v}); & 
2. R_{M}^{*} \text{ is an equivalence relation on } N^{*}; \\[4pt]
3. R_{L}^{*} \text{ is a universal relation on } N^{*}; & 
4. \text{If } w\sim v,\text{ then } R_{L}^{*}(\overline{w})=R_{L}^{*}(\overline{v}); \\[4pt]
5. \text{If } w\sim v,\text{ then } f^{*}(\overline{w},\psi)=f^{*}(\overline{v},\psi). & 
\end{array}
\]
\end{lemma}
\begin{proof}




We only need to prove Items 3, 4, and 5. Items 1 and 2 follow analogously
to Lemma \ref{lem:welldefinescp}, as the logical structure of $R_{M}^{*}$
relies on $\Box_{M}$-formulas which behave identically in the closure
$\Sigma_{\varphi}$ of both systems.

Item 3: Suppose $\overline{w},\overline{v}\in N^{*}$. By the definition
of $N^{*}$, there exist specific representatives $w'\in\overline{w}$
and $v'\in\overline{v}$ such that $w',v'\in N^{cg}$. As proven in
Lemma \ref{lem:universal-RL}, $R_{L}^{cg}$ is a universal relation
on $N^{cg}$, meaning $w'R_{L}^{cg}v'$ inherently holds. According
to the existential definition of $R_{L}^{*}$ in the filtration model,
the presence of such $w'$ and $v'$ directly yields $\overline{w}R_{L}^{*}\overline{v}$.
Thus, $R_{L}^{*}$ is strictly a universal relation on $N^{*}$.

Item 4: Since $R_{L}^{*}$ is a universal relation on $N^{*}$ (as
established in Point 4), the relational image of any element in $N^{*}$
is the entire set $N^{*}$. Therefore, if $w\sim v$, both $\overline{w}$
and $\overline{v}$ designate elements in $N^{*}$, trivially yielding
$R_{L}^{*}(\overline{w})=N^{*}=R_{L}^{*}(\overline{v})$.

Item 5: Suppose $w\sim v$. For any $(\psi>\chi)\in\Sigma_{\varphi}$,
since $w\sim v$, we have $\psi>\chi\in w\iff\psi>\chi\in v$. This
immediately ensures that the base sets are identical: $U_{\psi,w}=U_{\psi,v}$.
Crucially, by the definition of the subformula closure $\Sigma_{\varphi}$,
the presence of $(\psi>\chi)\in\Sigma_{\varphi}$ guarantees that
both modal preconditions $\Diamond_{M}\psi\in\Sigma_{\varphi}$ and
$\Diamond_{L}\psi\in\Sigma_{\varphi}$ hold. Since $w\sim v$, $w$
and $v$ must categorically agree on the presence of $\Diamond_{M}\psi$
and $\Diamond_{L}\psi$. This logical synchronization guarantees that
$w$ and $v$ trigger the exact same conditional branch in the piecewise
definition of $f^{*}$. Because $U_{\psi,w}=U_{\psi,v}$, $R_{M}^{*}(\overline{w})=R_{M}^{*}(\overline{v})$
(from Point 1), and $N^{*}$ as well as $W^{*}$ are model constants,
the subsequent intersection operations across any of the three branches
will mathematically yield identical results. Therefore, $f^{*}(\overline{w},\psi)=f^{*}(\overline{v},\psi)$.
\end{proof}
\begin{lemma}[Filtration Lemma] \label{lem:filtration-lemma-scp1}
For all $\psi\in\Sigma_{\varphi}$ and all $w\in W^{cg}$: 
\[
M^{cg},w\Vdash\psi\iff M^{*},\overline{w}\Vdash\psi
\]
\end{lemma}

\begin{proof}

For any $\overline{w}\notin N^{*}$ and $\psi\in\Sigma_{\varphi}$,
by truth conditions, $M^{cg},w\Vdash\psi\iff w\in V^{cg}(\psi)\iff\overline{w}\in V^{*}(\psi)\iff M^{*},\overline{w}\Vdash\psi$.

For $\overline{w}\in N^{*}$, we proceed by induction on the complexity
of $\psi$. We only prove the cases $\psi=\Box_{L}\gamma$ and $\psi=\alpha>\beta$;
other cases refer to Lemma \ref{lem:filtration-lemma}.

Case $\psi=\Box_{L}\gamma$:

($\Rightarrow$) Suppose $M^{cg},w\Vdash\Box_{L}\gamma$. Since $R_{L}^{*}$
is a universal relation on $N^{*}$ (by Lemma \ref{lem:welldefinescp1}),
we must show that for any $\overline{v}\in N^{*}$, $M^{*},\overline{v}\Vdash\gamma$.
Because $\Box_{L}\gamma\in w$ and $R_{L}^{cg}$ is a universal relation
on $N^{cg}$, for any $v\in N^{cg}$, $\gamma\in v$. Given $N^{*}=\{\overline{w}\mid w\in N^{cg}\}$
and IH, for any $\overline{v}\in N^{*}$, $M^{*},\overline{v}\Vdash\gamma$.

($\Leftarrow$) Suppose $M^{cg},w\nVdash\Box_{L}\gamma$. Then there
exists $v\in N^{cg}$ such that $v\nVdash\gamma$. Since $\overline{v}\in N^{*}$
and $R_{L}^{*}$ is a universal relation, $\overline{w}R_{L}^{*}\overline{v}$
holds. By IH, $M^{*},\overline{v}\nVdash\gamma$.
Thus, $M^{*},\overline{w}\nVdash\Box_{L}\gamma$.

Case $\psi=\alpha>\beta$:

$(\Rightarrow)$ Suppose $M^{cg},w\Vdash\alpha>\beta$. By the Truth
Lemma on $M^{cg}$, $\alpha>\beta\in w$. We need to show that for
any $\overline{v}\in f^{*}(\overline{w},\alpha)$, $M^{*},\overline{v}\Vdash\beta$.
By the definition of $f^{*}$ in the filtration model, regardless
of whether $\Diamond_{M}\alpha$ or $\Diamond_{L}\alpha$ belongs
to $w$, it must hold that $\overline{v}\in U_{\alpha,w}$. By the
definition of $U_{\alpha,w}$, for all $(\alpha>\chi)\in\Sigma_{\varphi}$,
if $\alpha>\chi\in w$, then $\chi\in v$. Since $\alpha>\beta\in\Sigma_{\varphi}$
and $\alpha>\beta\in w$, we necessarily have $\beta\in v$. By the
induction hypothesis, $M^{*},\overline{v}\Vdash\beta$. Thus, $f^{*}(\overline{w},\alpha)\subseteq[[\beta]]^{*}$,
which implies $M^{*},\overline{w}\Vdash\alpha>\beta$.

$(\Leftarrow)$ Suppose $M^{cg},w\nVdash\alpha>\beta$. By the Truth
Lemma on $M^{cg}$, $\alpha>\beta\notin w$. Based on the definition
of the selection function $f^{cg}$ in the generated submodel $M^{cg}$,
there exists a world $u\in f^{cg}(w,\alpha)$ such that $\beta\notin u$.
We consider its equivalence class $\overline{u}$ and prove $\overline{u}\in f^{*}(\overline{w},\alpha)$
by examining three cases according to the definition of $f^{*}$:

If $\Diamond_{M}\alpha\in w$: According to the (soc-m) constraint
in Lemma \ref{lem:cg-constraints-scp1}, $f^{cg}(w,\alpha)\subseteq R_{M}^{cg}(w)$.
Since $u\in f^{cg}(w,\alpha)$, it follows that $u\in R_{M}^{cg}(w)$.
Because $R_{M}^{cg}$ is an equivalence relation on $N^{cg}$, $w$
and $u$ agree on all $\Box_{M}$-formulas. Therefore, for any $\Box_{M}\gamma\in\Sigma_{\varphi}$,
$\Box_{M}\gamma\in w\iff\Box_{M}\gamma\in u$. By definition, this
means $\overline{w}R_{M}^{*}\overline{u}$. Meanwhile, since $\chi\in u$
holds for all $\alpha>\chi\in w$, we have $\overline{u}\in U_{\alpha,w}$.
Consequently, $\overline{u}\in U_{\alpha,w}\cap R_{M}^{*}(\overline{w})=f^{*}(\overline{w},\alpha)$.

If $\Diamond_{M}\alpha\notin w$ and $\Diamond_{L}\alpha\in w$: According
to the properties of $f^{cg}$ and the (sic) constraint in Lemma \ref{lem:cg-constraints-scp1},
$f^{cg}(w,\alpha)\subseteq N^{cg}$. Since $u\in f^{cg}(w,\alpha)$,
we immediately obtain $u\in N^{cg}$. By the definition of $N^{*}$,
it follows that $\overline{u}\in N^{*}$. Combined with $\overline{u}\in U_{\alpha,w}$,
we get $\overline{u}\in U_{\alpha,w}\cap N^{*}=f^{*}(\overline{w},\alpha)$.

If $\Diamond_{L}\alpha\notin w$: Since $u\in W^{cg}$, we have $\overline{u}\in W^{*}$.
Combined with $\overline{u}\in U_{\alpha,w}$, we get $\overline{u}\in U_{\alpha,w}\cap W^{*}=f^{*}(\overline{w},\alpha)$.

In all cases, we have established $\overline{u}\in f^{*}(\overline{w},\alpha)$.
Because $\beta\notin u$, IH yields $M^{*},\overline{u}\nVdash\beta$.
Therefore, $M^{*},\overline{w}\nVdash\alpha>\beta$.
\end{proof}

\begin{lemma} \label{lem:filter-constraints-scp1} The filtration
model satisfies the model constraints; that is, for all $\psi>\chi\in\Sigma_{\varphi}$
and $\overline{w}\in N^{*}$: 
\[
\begin{array}{@{}ll@{\qquad}ll@{}}
\text{(rl)} & R_{L}^{*}(\overline{w})=N^{*} & \text{(id)} & f^{*}(\overline{w},\psi)\subseteq[[\psi]]^{*}\\
\text{(ne)} & f^{*}(\overline{w},\psi)\neq\varnothing & \text{(soc-m)} & \text{If }M^{*},\overline{w}\Vdash\Diamond_{M}\psi,\ \text{then }f^{*}(\overline{w},\psi)\subseteq R_{M}^{*}(\overline{w})\\
\text{(modal)} & R_{M}^{*}(\overline{w})\subseteq R_{L}^{*}(\overline{w}) & \text{(sic)} & \text{If }M^{*},\overline{w}\Vdash\Diamond_{L}\psi,\ \text{then }f^{*}(\overline{w},\psi)\subseteq N^{*}
\end{array}
\]
\end{lemma}
\begin{proof}

\textbf{(rl)}: As proved in Lemma \ref{lem:welldefinescp1}, $R_{L}^{*}(\overline{w})=N^{*}$.

\textbf{(id)}: For any of the three cases in the definition of $f^{*}$,
it algebraically holds that $f^{*}(\overline{w},\psi)\subseteq U_{\psi,w}$.
Since $\psi>\psi\in\Sigma_{\varphi}$, for any $\overline{u}\in U_{\psi,w}$,
the construction condition implies $\psi\in u$. By the Filtration
Lemma, $M^{*},\overline{u}\Vdash\psi$, meaning $U_{\psi,w}\subseteq[[\psi]]^{*}$.
Therefore, $f^{*}(\overline{w},\psi)\subseteq[[\psi]]^{*}$.

\textbf{(ne)}: Since $f^{cg}(w,\psi)\neq\varnothing$, there exists
$v\in f^{cg}(w,\psi)$. By the definition of the canonical selection
function, for every $(\psi>\chi)\in\Sigma_{\varphi}$, we have $\psi>\chi\in w\implies\chi\in v$,
hence $\overline{v}\in U_{\psi,w}$. We then consider three cases:

If $\Diamond_{M}\psi\in w$, then $f^{cg}(w,\psi)\subseteq R_{M}^{cg}(w)$,
so $v\in R_{M}^{cg}(w)$. Since $R_{M}^{cg}$ is an equivalence relation
on $N^{cg}$, $w$ and $v$ agree on all $\Box_{M}$-formulas in $\Sigma_{\varphi}$,
hence $\overline{v}\in R_{M}^{*}(\overline{w})$. Therefore $\overline{v}\in U_{\psi,w}\cap R_{M}^{*}(\overline{w})=f^{*}(\overline{w},\psi)$.

If $\Diamond_{M}\psi\notin w$ and $\Diamond_{L}\psi\in w$, then
$f^{cg}(w,\psi)\subseteq N^{cg}$. Thus $v\in N^{cg}$, and by definition
$\overline{v}\in N^{*}$. Therefore $\overline{v}\in U_{\psi,w}\cap N^{*}=f^{*}(\overline{w},\psi)$.

If $\Diamond_{L}\psi\notin w$, then $f^{*}(\overline{w},\psi)=U_{\psi,w}\cap W^{*}$.
Since $\overline{v}\in W^{*}$, $\overline{v}\in U_{\psi,w}\cap W^{*}=f^{*}(\overline{w},\psi)$.In
all cases, we find $\overline{v}\in f^{*}(\overline{w},\psi)$, meaning
$f^{*}(\overline{w},\psi)\neq\varnothing$.

\textbf{(soc-m)}: Suppose $M^{*},\overline{w}\Vdash\Diamond_{M}\psi$.
By the Filtration Lemma, $\Diamond_{M}\psi\in w$. According to the
first case of the definition of $f^{*}$, $f^{*}(\overline{w},\psi)=U_{\psi,w}\cap R_{M}^{*}(\overline{w})$.
It directly follows that $f^{*}(\overline{w},\psi)\subseteq R_{M}^{*}(\overline{w})$.

\textbf{(modal)}: Since $R_{M}^{*}(\overline{w})$ is defined as an
equivalence relation on $N^{*}$, we naturally have $R_{M}^{*}(\overline{w})\subseteq N^{*}$.
As proven in Lemma \ref{lem:welldefinescp1}, $R_{L}^{*}(\overline{w})=N^{*}$,
thus $R_{M}^{*}(\overline{w})\subseteq R_{L}^{*}(\overline{w})$.

\textbf{(sic)}: Suppose $M^{*},\overline{w}\Vdash\Diamond_{L}\psi$.
By the Filtration Lemma, $\Diamond_{L}\psi\in w$. We consider the
two applicable cases in the definition of $f^{*}$:If $\Diamond_{M}\psi\in w$,
then $f^{*}(\overline{w},\psi)=U_{\psi,w}\cap R_{M}^{*}(\overline{w})$.
Since $R_{M}^{*}(\overline{w})\subseteq N^{*}$, we immediately have
$f^{*}(\overline{w},\psi)\subseteq N^{*}$.If $\Diamond_{M}\psi\notin w$,
then $f^{*}(\overline{w},\psi)=U_{\psi,w}\cap N^{*}$, which obviously
implies $f^{*}(\overline{w},\psi)\subseteq N^{*}$. In both scenarios,
the constraint is perfectly satisfied.
The case when $\psi$ is not an antecedent of a conditional in $\Sigma_\varphi$ can be proved analogously.
\end{proof}

\begin{theorem} $\mathbf{SCP1}$ has finite model property and is decidable. \end{theorem}

\section{Application}

To clearly demonstrate the theoretical advantages of the stratified semantic framework, this section uses specific case analyses to illustrate how the $\mathbf{SCP}$ and $\mathbf{SCP1}$ systems formally and semantically distinguish between ordinary counterfactuals, anti-metaphysical (anti-essential) conditionals, and anti-logical conditionals. Furthermore, it highlights the necessity of $\mathbf{SCP1}$ in resolving the problem of vacuous truth.

Let $q$ be a metaphysical necessity, such as ``water is $H_2O$''. In the actual world $w \in N$, $\Box_{M}q$ is true.

\paragraph{Case 1: Ordinary Counterfactuals}
Let $p_1$ be ``I missed the bus''. This is a metaphysically possible hypothesis, meaning $M, w \Vdash \Diamond_{M}p_1$.

According to the $\mathbf{SCP}$ semantic constraint (soc-m), if $w \in [\Diamond_{M}p_1]$, the selection function is restricted to the metaphysically accessible domain, i.e., $f(w, p_1) \subseteq R_{M}(w)$. Since $\Box_{M}q$ holds, water is $H_2O$ in all worlds within $R_{M}(w)$. Therefore, $f(w, p_1) \subseteq [q]$ necessarily holds, and the conditional $p_1 > q$ (``If I missed the bus, water would still be $H_2O$'') is true. Metaphysical necessity is strictly preserved under ordinary counterfactual hypotheses.

\paragraph{Case 2: Anti-Metaphysical Conditionals}
Let $p_2$ be ``water is XYZ''. This is a metaphysically impossible but logically possible hypothesis, meaning $M, w \Vdash \neg\Diamond_{M}p_2 \land \Diamond_{L}p_2$.

Here, $\Diamond_{M}p_2$ is false, so the constraint (soc-m) is not triggered, and the selection function $f(w, p_2)$ must reach beyond $R_{M}(w)$. However, because the antecedent is logically possible ($\Diamond_{L}p_2$ is true), the constraint (sic) is activated, requiring $f(w, p_2) \subseteq N$. The selection function extends beyond the metaphysical domain into the logic-normal world $N$, where specific metaphysical laws may fail. In this broader space, $q$ is no longer universally true. This precisely corresponds to the $\mathbf{SCP}$ proposition $\Box_{M}\psi \nvDash \Diamond_{L}\varphi \to (\varphi > \psi)$. The system successfully allows us to suspend metaphysical laws while still using classical logic for meaningful deduction within $N$.

\paragraph{Case 3: Anti-Logical Conditionals}
Let $p_3$ be ``Hobbes squared the circle'', which is a logical impossibility, meaning $M, w \Vdash \neg\Diamond_{L}p_3$. Let $r$ be ``sick children in Afghanistan would be cured''.

Since $\Diamond_{L}p_3$ is false, the constraint (sic) protecting classical logical deduction is deactivated. The selection function $f(w, p_3)$ cannot find a world in $N$ where the antecedent is true, and thus must point to the anti-logical world set $I$. In Case 2, although the evaluation deviates from actuality, the worlds remain in $N$, enjoying the recursively defined truth conditions of classical logic. In Case 3, however, the worlds fall into $I$, where the truth values of complex formulas (such as conjunctions and disjunctions) are no longer evaluated recursively. Logical laws completely collapse here.

Although $\mathbf{SCP}$ establishes a refined stratified structure $N/R_M/I$, as a foundational framework, it permits the selection function $f$ to be empty. This exposes a theoretical vulnerability when handling Case 3. If we construct a specific $\mathbf{SCP}$ model where no world in the anti-logical domain $I$ satisfies ``Hobbes squared the circle'' (i.e., $f(w, p_3) = \varnothing$), then since the empty set is trivially a subset of any set, the conditional $p_3 > r$ would be vacuously true simply because its antecedent cannot be satisfied. This means $\mathbf{SCP}$ might still compromise with vacuism in certain local models.

The extended system $\mathbf{SCP1}$ introduces the non-emptiness constraint (ne), forcing $f(w, \varphi) \neq \varnothing$ for all normal worlds $w$ and formulas $\varphi$. When faced with $p_3$, $\mathbf{SCP1}$ forces the model to presuppose a sufficiently rich domain $I$ and select at least one world containing this logical contradiction for evaluation. In these selected anti-logical worlds, due to the lack of recursive logical constraints, $f(w, p_3)$ will not automatically be a subset of $[r]$ unless there is a substantive connection between ``squaring the circle'' and ``curing Afghan children'' in the specific valuation of the model. Consequently, $\mathbf{SCP1}$ can robustly evaluate $p_3 > r$ as false.

\section{Conclusion}

This paper develops a formal theory to distinguish impossibility through
the stratified counterpossible logic $\mathbf{SCP}$ and its extension
$\mathbf{SCP1}$. We establish the soundness, completeness, finite
model property, and decidability for both axiomatic systems.

To address the problem of vacuous truth, $\mathbf{SCP1}$ requires
the selection function to be non-empty for all antecedents, including
logical contradictions. By presupposing a robust set of anti-logical
worlds $I$, the system ensures that the truth values of counterpossibles
depend on the substantive connection between the antecedent and the
consequent.

The hierarchical framework proposed here is generalizable beyond metaphysical
and logical modalities. The relations $R_{M}$ and $R_{L}$ can be
applied to other modal hierarchies, such as distinguishing primary
from secondary obligations in deontic logic, or modeling cognitive
blind spots in epistemic logic. This layered approach offers a semantic
path for resolving modal conflicts and characterizing cognitive limitations.

Future research should further refine the internal structure of impossible
worlds. While this paper treats the anti-logical domain $I$ as a
space with minimal constraints, future work could classify $I$ using
specific non-classical logics, such as paraconsistent systems. This
would allow for a more precise characterization of the intermediate
domain between metaphysical impossibility and absolute logical inconsistency.

\paragraph*{Acknowledgments.} This work was supported by the MOE Project of Key Research Institute of Humanities and Social Sciences at Universities (Grant No. 22JJD520001).

\bibliographystyle{eptcs}
\nocite{*}
\bibliography{generic}

\end{document}